\documentclass[11pt,letterpaper]{article}

\usepackage{microtype}
\usepackage{amsmath, amsthm}
\usepackage{graphicx}
\usepackage{authblk}
\usepackage{hyperref}
\usepackage[margin=1in]{geometry}

\newtheorem{theorem}{Theorem}
\newtheorem{lemma}[theorem]{Lemma}
\newtheorem{corollary}[theorem]{Corollary}
\newtheorem{definition}[theorem]{Definition}

\usepackage{comment}  

\title{Subcubic Equivalences Between Graph Centrality Measures and Complementary Problems\footnote{Supported in part by NSF CAREER award 1053605, NSF grant CCF-1161626, ONR YIP award N000141110662, and a DARPA/AFOSR grant FA9550-12-1-0423.}
}

\author[1]{Mahdi Boroujeni}
\author[2]{Sina Dehghani}
\author[3]{Soheil Ehsani}
\author[3]{MohammadTaghi HajiAghayi}
\author[3]{Saeed Seddighin}
\affil[1]{Department of Computer Engineering, Sharif University of Technology \url{safarnejad@ce.sharif.edu}}
\affil[2]{School of Mathematics, Institute for Research in Fundamental Sciences (IPM) \url{dehghani@ipm.ir}}
\affil[3]{Department of Computer Science, University of Maryland \url{ehsani, hajiagha, sseddigh@umd.edu}}
\date{}

\newcommand{\APSP}{\ensuremath{\mathsf{APSP}}}
\newcommand{\APSPV}{\ensuremath{\mathsf{APSPVerification}}}
\newcommand{\APSPVV}{\ensuremath{\mathsf{APSPVerificationIndex}}}
\newcommand{\COAPSPV}{\ensuremath{\mathsf{CoAPSPVerification}}}
\newcommand{\Radius}{\ensuremath{\mathsf{Radius}}}
\newcommand{\Center}{\ensuremath{\mathsf{Center}}}
\newcommand{\Median}{\ensuremath{\mathsf{Median}}}
\newcommand{\RadiusVertex}{\ensuremath{\mathsf{RadiusVertex}}}
\newcommand{\MedianV}{\ensuremath{\mathsf{MedianVertex}}}
\newcommand{\Diameter}{\ensuremath{\mathsf{Diameter}}}
\newcommand{\DiameterE}{\ensuremath{\mathsf{DiameterVertex}}}
\newcommand{\NEG}{\ensuremath{\mathsf{NegativeTriangle}}}
\newcommand{\NEGV}{\ensuremath{\mathsf{NegativeTriangleVertex}}}
\newcommand{\CONEG}{\ensuremath{\mathsf{CoNegativeTriangle}}}

\newcommand{\COCenterVertex}{\ensuremath{\mathsf{coCenter}}}
\newcommand{\CoRadius}{\ensuremath{\mathsf{CoRadius}}}
\newcommand{\CODiameter}{\ensuremath{\mathsf{CoDiameter}}}
\newcommand{\COMedian}{\ensuremath{\mathsf{CoMedian}}}
\newcommand{\COMedianVertex}{\ensuremath{\mathsf{coMedian}}}

\newcommand{\subless}{\le_{n^3}}
\newcommand{\subequal}{=_{n^3}}
\newcommand{\otilda}{\widetilde{O}}
\newcommand{\CORadius}{\ensuremath{\mathsf{CoRadius}}}
\newcommand{\poly}{\mathsf{poly}}

\newtheorem{observation}{Observation}[section]
\newenvironment{proofof}[1]{{\bf Proof of #1:  }}{\hfill\rule{2mm}{2mm}}

\begin{document}

\maketitle

\begin{abstract}
	Despite persistent efforts, there is no known technique for obtaining unconditional super-linear lower bounds for the computational complexity of the problems in P. Vassilevska Williams and Williams \cite{focs} introduce a fruitful approach to advance a better understanding of the computational complexity of the problems in P. In particular, they consider \textit{All Pairs Shortest Paths} (\APSP{}) and other fundamental problems such as checking whether a matrix defines a metric, verifying the correctness of a matrix product, and detecting  a negative triangle in a graph.

Abboud, Grandoni, and Vassilevska Williams \cite{soda} study well-known graph centrality problems such as \Radius{}, \Median{}, etc., and make a connection between their computational complexity to that of two fundamental problems, namely \APSP{} and \textit{\Diameter{}}. They
show any algorithm with subcubic running time for these centrality problems, implies a subcubic algorithm for either \APSP{} or \Diameter{}.

In this paper, we define vertex versions for these centrality problems and based on that we introduce new complementary problems. The main open problem of \cite{soda} is whether or not \APSP{} and \Diameter{} are equivalent under subcubic reduction. One of the results of this paper is \APSP{} and \CODiameter{}, which is the complementary version of \Diameter{}, are equivalent. Moreover, for some of the problems in this set, we show that they are equivalent to their complementary versions. Considering the slight difference between a problem and its complementary version, these equivalences give us the impression that every problem has such a property, and thus \APSP{} and \Diameter{} are equivalent. This paper is a step forward in showing a subcubic equivalence between \APSP{} and \Diameter{}, and we hope that the approach introduced in our paper can be helpful to make this breakthrough happen.


\end{abstract}

\section{Introduction}

Computational complexity focuses on classifying algorithmic problems mostly through providing lower bounds to show solving a certain problem requires at least a certain amount of time, memory/space, number of gates in a circuit, etc.
%
However, despite persistent efforts, still, there is no known technique for proving unconditional super-linear lower bounds for polynomially solvable problems.
Vassilevska Williams and Williams \cite{focs} introduce a fruitful approach to provide evidence that significantly improving the running time for solving a certain set of problems in P is unlikely.
Their approach is to use reductions to show improving upon a given upper bound for a computational problem, implies improving a breakthrough algorithm for another famous and fundamental problem. More specifically, consider a well-studied problem $A$ for which the best-known algorithm has running time $\otilda(n^c)$\footnote{The $\otilda$ notation suppresses $\text{polylog}$ terms in $n$ and $M$, where $M$ is an upper bound on the weights of the input graph.}. By providing a reduction from another problem $B$ to $A$, it can be shown that an $\otilda(n^{c'-\epsilon})$ time algorithm for problem $B$, for a constant $\epsilon>0$, yields an $\otilda(n^{c-\delta})$ time algorithm for problem $A$, for another constant $\delta>0$. This means it is unlikely to obtain an $\otilda(n^{c'-\epsilon})$ time algorithm for problem $B$. 
For $c,c' = 3$ a reduction of the above kind is called a \textit{subcubic\footnote{Sometimes these reductions and time complexities are called \textbf{truly} subcubic since they do not count subpolynomial improvements such as polylogarithmic factors on cubic times.} reduction}. Two problems $A$ and $B$ are called \textit{subcubic equivalent}, if there is a subcubic reduction from $A$ to $B$ and a subcubic reduction from $B$ to $A$ \cite{soda, focs}. 

Vassilevska Williams and Williams \cite{focs} prove a subcubic equivalence between \APSP{} and seven other fundamental problems, such as checking whether a matrix defines a metric, verifying the correctness of a matrix product over the (min, +)-semiring, and detecting if a weighted graph has a triangle of negative total edge weight. Since then several works used the same approach to obtain interesting hardness results for polynomially solvable problems 
such as edit distance~\cite{editdistance}, LCS~\cite{lcs}, a number of dynamic problems~\cite{abboud2014popular, patrascu2010towards}, RNA folding~\cite{valiant}, and tree edit distance~\cite{ted}.

In the past few decades, there has not been any significant improvement or computational lower bound for graph centrality problems, especially for \APSP{}. Therefore, proving a subcubic equivalence between a certain problem with cubic time and \APSP{} could be ``a huge and unexpected algorithmic breakthrough''~\cite{soda}. Floyd \cite{floyd} and Warshall \cite{warshall} proposed an $O(n^3)$ algorithm for \APSP{} in 1962. There have been many attempts to improve this running time.
Nonetheless, the best-known algorithm for \APSP{} runs in time $O(\frac{n^3}{2^{\Theta(\sqrt{\log n})}})$\footnote{This still is not $O(n^{3-\epsilon})$ for a positive constant $\epsilon$.} \cite{williams2014faster}. However, still ``One of the Holy Grails of the graph algorithms is to determine whether this cubic complexity is inherent, or whether a significant improvement (\textit{e.g.}, an $O(n^{2.99})$ time) is possible''~\cite{focs}.

Abboud, Grandoni, and Vassilevska Williams \cite{soda} study a series of fundamental graph centrality problems having lots of applications such as finding influential person(s) in social networks, finding key infrastructure nodes in the Internet or urban networks, and detecting super-spreaders of disease.
The problems they consider are \Radius{}, \Median{}, \Diameter{}, etc., for which the fastest known algorithms are of $\otilda(n^{3-o(1)})$ running time. Abboud \textit{et al.} \cite{soda} make a connection between the complexity of these problems to that of two fundamental problems, namely \APSP{} and \Diameter{}. In \Diameter{}, we are asked to find the maximum distance between any two nodes of a graph. 
They prove \APSP{}, \Radius{}, and \Median{} are equivalent under subcubic reductions, \textit{i.e.}, a subcubic algorithm for any of these problems implies a subcubic algorithm for all of the others. They also show \Diameter{}, reach centrality, and any constant factor approximation algorithm for betweenness centrality are equivalent under subcubic reductions. 
 
However, the main open question is whether we can obtain a similar connection between \Diameter{} and \APSP{}. It is straightforward to show a reduction from \Diameter{} to \APSP{}; Once you have all the distances between the nodes, you can find the maximum distance in time $O(n^2)$ but is there a subcubic reduction from \APSP{} to \Diameter{}? Or can the largest distance between vertices\footnote{Namely diameter.} of a graph be calculated faster\footnote{In a subcubic time.} than the time required to calculate all pairwise distances?

In this paper, we consider a complementary version of \Diameter{} and relate its computational complexity to \APSP{}. In particular, we define \CODiameter{} as the problem of finding a vertex of the input graph which is not an endpoint of a diameter and show a subcubic reduction from \APSP{} to \CODiameter{}.

\begin{theorem}\label{thm1.1}
	\APSP{} and \CODiameter{} are subcubic equivalent.
\end{theorem}

Furthermore, we define complementary problems for other fundamental problems studied before, such as \CORadius{} and \COMedian{}.
In this paper, we prove subcubic equivalences between \APSP{}, \COMedian{}, and \CORadius{}, which lead to subcubic equivalences between \Median{} and \COMedian{}, and \Radius{} and \CoRadius{}.

\begin{theorem}\label{thm1.2}
\APSP{}, \COMedian{}, and \CORadius{} are all subcubic equivalent.
\end{theorem}



We also make a connection between the computational complexities of \CONEG{} and \COAPSPV{} to that of the \Diameter{} problem.
In particular, the reduction from \CONEG{} to \Diameter{} is of special interest, since a reduction from \NEG{} to \Diameter{} would resolve the open problem of reducing \APSP{} to \Diameter{}. Moreover, the reduction from \CONEG{} to \NEG{} yields a common source for the hardness of both \APSP{} and \Diameter{}.
\begin{theorem}
	There exists a subcubic reduction from \CONEG{} to \Diameter{}.
\end{theorem}
\begin{theorem}
	There exists a subcubic reduction from \Diameter{} to \COAPSPV{}.
\end{theorem}
\begin{theorem}
	There exists a subcubic reduction from \CONEG{} to \NEG{}.
\end{theorem}

The number of the problems considered in this paper may be high; however, Figure \ref{fig:results} perfectly illustrates the time complexity relations between the problems mentioned above. Note that in Figure \ref{fig:results} any path from a problem $A$ to s problem $B$ denotes a subcubic reduction from problem $A$ to problem $B$. Prior reductions are shown via dotted arrows, except trivial reductions to \APSP{} which are not shown for the sake of clarity.

\begin{figure}
	\centering
	\includegraphics[width=9cm]{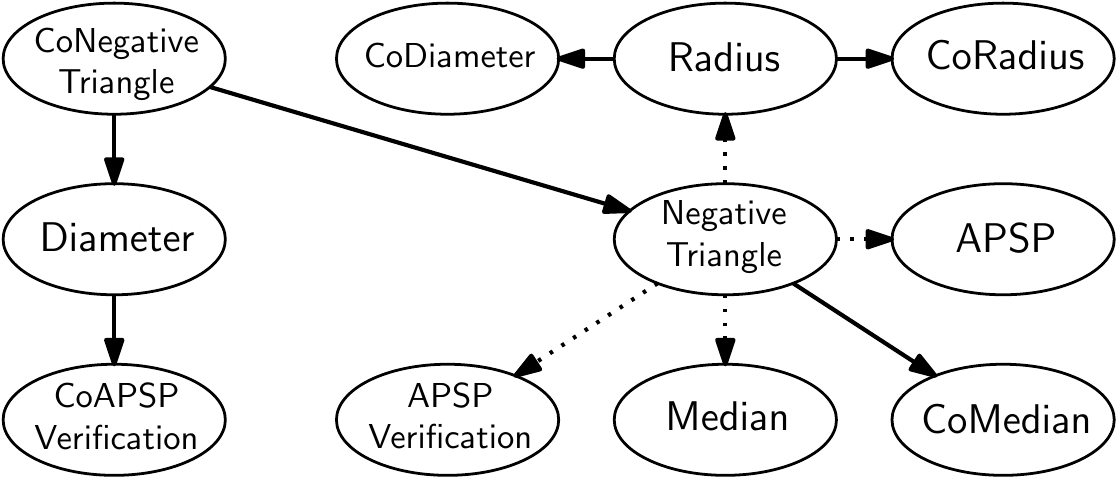}
	
	\caption{Dotted arrows show reductions prior to this work, and solid arrows illustrate the reductions that we present in this work. Note that we omit trivial reductions to \APSP{} here.}
\end{figure}\label{fig:results}

\section {Related Work}
The most related studies to this paper are by Vassilevska Williams and Williams \cite{focs} and Abboud, Grandoni, and Vassilevska Williams \cite{soda}. Vassilevska Williams \textit{et al.} \cite{focs} introduce the notion of subcubic reduction and prove subcubic equivalences between \APSP{}  and seven other fundamental problems. Abboud \textit{et al.} \cite{soda} use the same approach to obtain subcubic equivalences among \APSP{}, \Diameter{}, and graph centrality problems such as \Radius{} and \Median{}. They show any subcubic algorithm for graph centrality problems can be used as a black box to obtain a subcubic algorithm for \APSP{} or \Diameter{}.
Furthermore, Lincoln, Vassilevska Williams, and Williams study similar problems in sparse graphs~\cite{soda3}.

As mentioned above, \APSP{} is among the most well-studied 
problems in P, for which there has been a tremendous amount of work to improve its running time 
(see, \textit{e.g.}, \cite{chan2010more, fredman1976new, han2012n, takaoka2005n3loglogn, williams2014faster, zwick2005slightly}).
Williams \cite{williams2014faster} proves there exists an $O(\frac{n^3}{2^{\Theta(\sqrt{\log n})}})$ time algorithm for \APSP{}, which is the best-known algorithm so far. However, there are faster algorithms for graphs with small integer weights (see \cite{shoshan1999all, zwick2002all}).

The \textit{Strong Exponential Time Hypothesis (SETH)} of Impagliazzo, Paturi, and Zane \cite{impagliazzo1999complexity, impagliazzo1998problems}, has also been an extremely popular conjecture and a powerful tool to provide surprising lower bounds on different problems. According to SETH, there is no $O((2-\epsilon)^n \text{poly}(n))$ time algorithm to determine the satisfiability of an $n$-variable CNF formula for some positive $\epsilon$.
Roditty and Vassilevska Williams \cite{roditty2013fast} show lower bounds for approximating the diameter of sparse graphs using SETH. 
Abboud, Vassilevska Williams, and Wang study the complexity of computing different versions of diameter and radius of sparse graphs using a similar approach~\cite{soda2}.

\section{Problems}
In all of the problems that we study, we assume the given graph has $n$ vertices and $m$ edges. We refer to the vertex set and edge set of a graph $G$ by $V(G)$ and $E(G)$, respectively. For brevity, sometimes we omit the terms directed and weighted, but all of the graphs are considered to be both directed and weighted unless otherwise stated. Also, the weights of the edges are integer numbers between $-M$ and $M$ where $M$ is a large enough integer number that is polynomially bounded by $n$\footnote{The condition of $M = \text{poly}(n)$ is a standard assumption in the definition of \APSP{} and related problems~\cite{ipecsurvey}. However, in cases that $M$ is not polynomially bounded by $n$, a $\text{polylog}\,M$ factor is added to the running time which is natural~\cite{focs} and hidden under the $\otilda$ notation. Moreover, some of our reductions may introduce a $\text{polylog}\,M$ term which can be cut off using a randomized technique similar to that of~\cite{focs}.}. We assume all of the basic operations (addition, subtraction, multiplication, etc.) take $O(1)$ time. Whenever we use $\infty$, it represents a number larger than any other integer number including $M$. Similarly, $-\infty$ is always strictly less than any integer number including $-M$. If there is no edge between a pair of vertices, we assume an edge with weight $\infty$ for that pair. Addition and multiplication of positive numbers to $\infty$ and $-\infty$ result in $\infty$ and $-\infty$ respectively, except for multiplication by zero which is zero. 

We say a problem $A$ is \textit{subcubically not harder} than a problem $B$ or there is a subcubic reduction from $A$ to $B$ if every algorithm that solves problem $B$ in subcubic time can be used as a black box to solve problem $A$ in subcubic time. We denote this reduction with $A \subless B$. Similarly, two problems $A$ and $B$ are \textit{subcubically equivalent} if both $A \subless B$ and $B \subless A$ hold. This relation is referred to by $A \subequal B$.

In the following, we define all of the problems in detail and explain the relation between them. We divide the problems into three different categories. The first category contains the problems in which the objective function is to measure a quantity of a given graph. In the second category, we define the vertex version of the same problems. Finally, in the third category, we define the complementary version of the problems based on their vertex version.
The definitions of the problems may seem repetitive, but as we show later in the paper, this does not necessarily mean the problems are equivalent.

\subsection{The First Category: Original Version}

The problems of this category are some of the well-studied cubic-time problems in their standard definition. In the following, we shortly bring a definition of each problem so that the reader has a reference to compare these problems with the problems of the next categories.

\begin{definition}
	Given a graph $G$, \APSP{} asks for an $n \times n$ matrix $D$ such that $D_{i,j}$ specifies the distance between the $j$'th vertex from $i$'th vertex of $G$.
\end{definition}
We also study another variant of the \APSP{} problem in which we are not required to compute the whole matrix of distances, but we only need to verify if a given matrix is the correct distance matrix of the graph.

\begin{definition}
	Given a graph $G$ and a matrix $D$, the objective of \APSPV{} is to determine whether $D$ is the correct distance matrix of $G$.
\end{definition}

One of the important problems that have been studied in the literature of subcubic equivalences is the \NEG{} problem. In this problem, the goal is to determine whether a given graph has a triangle with negative total weight. Although the solution of every instance of this problem is either YES or NO, it has been shown that this problem is as hard as \APSP{} with regard to having a subcubic algorithm \cite{focs}.

\begin{definition}
	Given a graph $G$, \NEG{} asks whether the graph has a triangle with negative total weight.
\end{definition}
We also study the \Median{}, \Radius{}, and \Diameter{} problems for weighted graphs with non-negative weights. All these problems have been vastly studied in the literature. Many algorithms have been proposed for each of these problems, but none of them has a subcubic runtime \cite{hakimi1964optimum,aingworth1999fast}. In a recent work of Abboud et al.~\cite{soda}, it has been shown that a subcubic algorithm for either of these problems leads to a subcubic algorithm for the \APSP{} problem. It is trivial to show that any subcubic algorithm for \APSP{} solves any of these problems in subcubic time.
\begin{definition}
	Given a graph $G$ with non-negative edge weights, the goal of \Radius{} is to find the smallest number $R^*$, such that there exists a $v\in G$ that can reach every other vertex within a distance of $R^*$.
\end{definition}
\begin{definition}
	Given a graph $G$ with non-negative edge weights, the goal of \Median{} is to find a vertex whose total sum of distances to all other vertices is minimum and report this total sum.
\end{definition}

\begin{definition}
	Given a graph with non-negative edge weights, \Diameter{} asks to compute the longest distance between any pair of vertices in $G$.	
\end{definition}

\subsection{The Second Category: Vertex Version}

In the second category, we introduce the vertex versions of the problems in the first category. In Lemma \ref{mosavian} we show equivalences between the original version and the vertex version for some of the problems\footnote{A similar idea can be used to prove the same claim for the rest of the problems.}.

\begin{definition}
	Given a graph $G$ and a matrix $D$, the goal of the \APSPVV{} problem is to either report that $D$ is the correct distance matrix of $G$ or return an index $(i,j)$ such that $D_{i,j}$ is \textbf{not} equal to the distance between the $j$'th vertex from the $i$'th vertex.
\end{definition}
\begin{definition}
	Given a graph $G$, the goal of the \NEGV{} problem is to either report the graph has no triangle with negative total weight or report a vertex which forms such a triangle with two other vertices.
\end{definition}

\begin{definition}
	Given a graph $G$ with non-negative edge weights, the goal of the \RadiusVertex{} problem is to find a vertex which has the minimum of maximum distance to all other vertices.
\end{definition}

Note that \RadiusVertex{} is equivalent to finding a center of $G$.

\begin{definition}\sloppy
	Given a graph $G$ with non-negative edge weights, the goal of the \MedianV{} problem is to find a vertex which has the minimum total sum of distances to all other vertices.
\end{definition}

\begin{definition}\sloppy
	Given a graph $G$ with non-negative edge weights, the goal of the \DiameterE{} problem is to find a vertex $u$ such that there exists a vertex $v$ that has a distance from $u$ equal to the diameter of the graph.
\end{definition}

The reason we define different versions of a problem is that this helps convey a better understanding of the idea behind our reductions. It is important to mention that these different definitions of a problem do not change its hardness under subcubic reductions. To prove this, we use binary search as the primary tool to solve one problem from another. In other words, it can be shown that each problem in the first category is equivalent to its corresponding problem of the second category.


\begin{lemma}\label{mosavian}
	Given a graph $G=(V, E)$, the following pairs of problems are equivalent under subcubic reduction.
	\begin{itemize}
		\item \Radius{} and \RadiusVertex{} $\mathsf{(}$\Center{}$\mathsf{)}$.
		\item \Median{} and \MedianV{}.
		\item \Diameter{} and \DiameterE{}.
	\end{itemize}
\end{lemma}

The proof of this lemma is in Appendix \ref{trivialapp}.
\subsection{The Third Category: Complementary Version}

The problems of this category are defined in the same way as the problems of the second category; however, the objective here is exactly the opposite. For instance, in \DiameterE{} the goal is to find an endpoint of a diameter, where the goal of \CODiameter{} is to find a vertex that is not an endpoint of a diameter. Although the definitions of two problems seem very similar, we point out a wide gap between them. This is interesting since as shown in this paper, some of the other similar problems such as \Median{} and \Radius{}, are equivalent to their complementary versions. As a contribution of this paper, we simplify the gap between \Diameter{} and \APSP{} to the gap between \Diameter{} and \CODiameter{}.


\begin{definition}
	Given a graph $G$ and a matrix $D$, the goal of \COAPSPV{} is to either report that none of the entries of $D$ is correct or report a pair $(i,j)$ such that $D_{i,j}$ is equal to the distance between vertex $j$ from vertex $i$ of $G$.
\end{definition}

\begin{definition}
	Given a graph $G$, the goal of \CONEG{}  is to either report that every vertex of $G$ is in negative triangle or return a vertex which is \textbf{not} in any negative triangle.
\end{definition}

\begin{definition}
	Given a graph $G$, the goal of \CoRadius{} is to report a vertex which is \textbf{not} a solution to the \RadiusVertex{} problem for the same input, if exists one. Otherwise, reports that every vertex is a solution to \Radius{}, i.e. all vertices are centers of $G$.
\end{definition}

\begin{definition}
	Given a graph $G$, the goal of \COMedian{} is to report a vertex which is \textbf{not} a solution to \Median{}, if exists one.
\end{definition}

\begin{definition}
	Given a graph $G$, the goal of \CODiameter{} is to report a vertex which is \textbf{not} a solution to \DiameterE{}, if exists one.
\end{definition}

\section{Reductions}
In this section, we explain our reductions in detail. In Section \ref{sec1} we provide a subcubic reduction from \Radius{} to \CoRadius{} and \CODiameter{}. In Section \ref{sec2} we show a subcubic reduction from \NEG{} to \COMedian{}. Next, in Sections \ref{sec3} and \ref{sec4} we demonstrate subcubic reductions from \Diameter{} to \COAPSPV{} and from \CONEG{} to \Diameter{}, respectively. Finally, in Section \ref{sec5}, we reduce \CONEG{} to \NEG{}.
\subsection{\Radius{} to \CORadius{} and \Radius{} to \CODiameter{}}\label{sec1}
The main idea behind our proof is constructing a new graph instance and provide a subcubic reduction via a binary search.
\begin{lemma}\label{radToCorad}
	Given an $\otilda(T(n))$ time algorithm for \CORadius{}, where $T(n)$ is polynomial in $n$, there exists an $\otilda(T(n)+n^2)$ time algorithm for \Radius{}.
\end{lemma}

\begin{proof}
	First, without loss of generality, we assume every edge in $G$ has an even weight since otherwise, we can double the weight of each edge. Let $\mathcal{A}$ be an $\otilda(T(n))$ time algorithm for \CORadius{}. Given graph $G$, we construct a graph $G'$ as follows. Put all vertices and edges of $G$ in $G'$, plus two new vertices $x$ and $y$. For each vertex $v \in V(G')\setminus\{x, y\}$ add four edges from $v$ to $x$ and $y$, and from $x$ and $y$ to $v$ each with weight $q$. 
	Now we claim that the radius of $G$ is less than $2q$ if and only if there is a vertex in $G'$ which is not a center. For simplicity, we call such a vertex a \COCenterVertex{}.
	
	Given the claimed proposition we can use algorithm $\mathcal{A}$ to determine whether there exists a \COCenterVertex{} in $G'$, in time $\otilda(T(n))$. Hence, a binary search on $q$ can find the minimum value of $q$ such that the radius of $G$ is no less than $2q$, \textit{i.e.}, every vertex in $G'$ is a center. The number of times we need to use $\mathcal{A}$ is $O(\log nM) \in \otilda(1)$; therefore, there exists an $\otilda(T(n))$ time algorithm for \Radius{}.
	
	To prove the claim, we first show that if there exists a \COCenterVertex{} in $G'$, then the radius of $G$ is less than $2q$. Let $u$ be such a vertex. Note that, the distance between $u$ and any other vertex is at most $2q$ by the construction. Since $u$ is not a center, there exists a  center $v \in V(G')$ such that the distance between $v$ and any other vertex is less than $2q$. Thus, the radius of $G$ is less than $2q$.
	
	Similarly, if the radius of $G$ is less than $2q$, then vertex $x$ is a \COCenterVertex{} in $G'$ since the shortest path between $x$ and $y$ is of length $2q$.
	Therefore, there exists an $\otilda(T(n)+n^2)$ time algorithm for \Radius{}.
\end{proof}

Interestingly, in proof of Lemma \ref{radToCorad}, we do not need to know which vertex is a \COCenterVertex{}. More precisely, it is only sufficient to know whether all vertices of the graph are centers or not. Via the following observation, we conclude the same proof can be used to reduce \Radius{} to \CODiameter{}.
\begin{observation}\label{obvious}
Every graph $G$ has a vertex $u$ which is a \COCenterVertex{} if and only if it has a vertex $v$ which is not a diameter endpoint of the graph.
\end{observation}

In other words, deciding whether the radius and the diameter of $G$ are equal is equivalent to \APSP{} under subcubic reduction.

\begin{corollary}\label{corolla}
	Given an $\otilda(T(n))$ time algorithm for \CODiameter{}, there exists an $\otilda(T(n)+n^2)$ time algorithm for \Radius{}.
\end{corollary}

The following theorems follow directly from Lemma \ref{radToCorad} and Corollary \ref{corolla}.
\begin{theorem}\label{thms1}
	$\Radius \subless \CoRadius{}$.
\end{theorem}
\begin{theorem}\label{thms2}
	$\Radius \subless \CODiameter{}$.
\end{theorem}

Note that due to Abboud et al.~\cite{soda} \APSP{} and \Radius{} are equivalent under subcubic reduction. Thus, by Theorems \ref{thms1} and \ref{thms2}, \CORadius{} and \CODiameter{} are also equivalent to \APSP{} under subcubic reductions.

\begin{corollary}\label{apspcoradcodiam}
	$\APSP{} \subequal \CoRadius{} \subequal \CODiameter{}$.
\end{corollary}
\subsection{\NEG{} to \COMedian{}}\label{sec2}
In this section, we provide a subcubic reduction from \NEG{} to \COMedian{}. The reduction uses a tricky graph construction to create a symmetric instance graph that helps to make a connection from \NEG{}, which is subcubically equivalent to \APSP{}, to \COMedian{}.
\begin{lemma}\label{lemlemeto}
	Given an $\otilda(T(n))$ time algorithm for \COMedian{}, where $T(n)$ is polynomial in $n$, there exists an $\otilda(T(n)+n^2)$ time algorithm for \NEG{}.
\end{lemma}

\begin{proof}
	Given a graph $G(V,E,w)$, we construct a graph $G'(V',E',w')$ with $3$ times as many vertices as $G$. The approach to solving the problem is to see whether $G'$ has a vertex which is not median (\COMedianVertex{}). If not, we show finding out whether $G$ has a negative triangle, given the fact that $G'$ has no negative edge, can be done merely by running \textit{Dijkstra's} algorithm from an arbitrary vertex of $G'$.
	
	Without loss of generality, we assume there exists an edge between every two vertices of $G$; otherwise, we put an edge with a large enough even weight $H$ and be sure that it does not contribute to any negative triangle. The vertex set of $G'$ contains three copies of $V(G)$, namely $A$, $B$ and $C$. Let $v_X$ denote a copy of a vertex $v\in V(G)$ in part $X\in\{A,B,C\}$ of $G'$. We draw an edge of weight $H/2$ from every $v_X$ to every other $u_X$ in the same part. Moreover, for every two vertices $v_A$ and $u_B$ we draw an edge of weight $H+w_{v,u}$ from $v_A$ to $u_B$ and an edge of weight $2H-w_{v,u}$ from $u_B$ to $v_A$. Moreover, we assume $w_{v,v}=H$. We do the same for edges between parts $B$ and $C$ and parts $C$ and $A$. Figure \ref{fig:Gp} shows graph $G'$ and the symmetry between its three parts.
	
	\begin{figure}[h]{}
		\centering
		\includegraphics[scale=0.36]{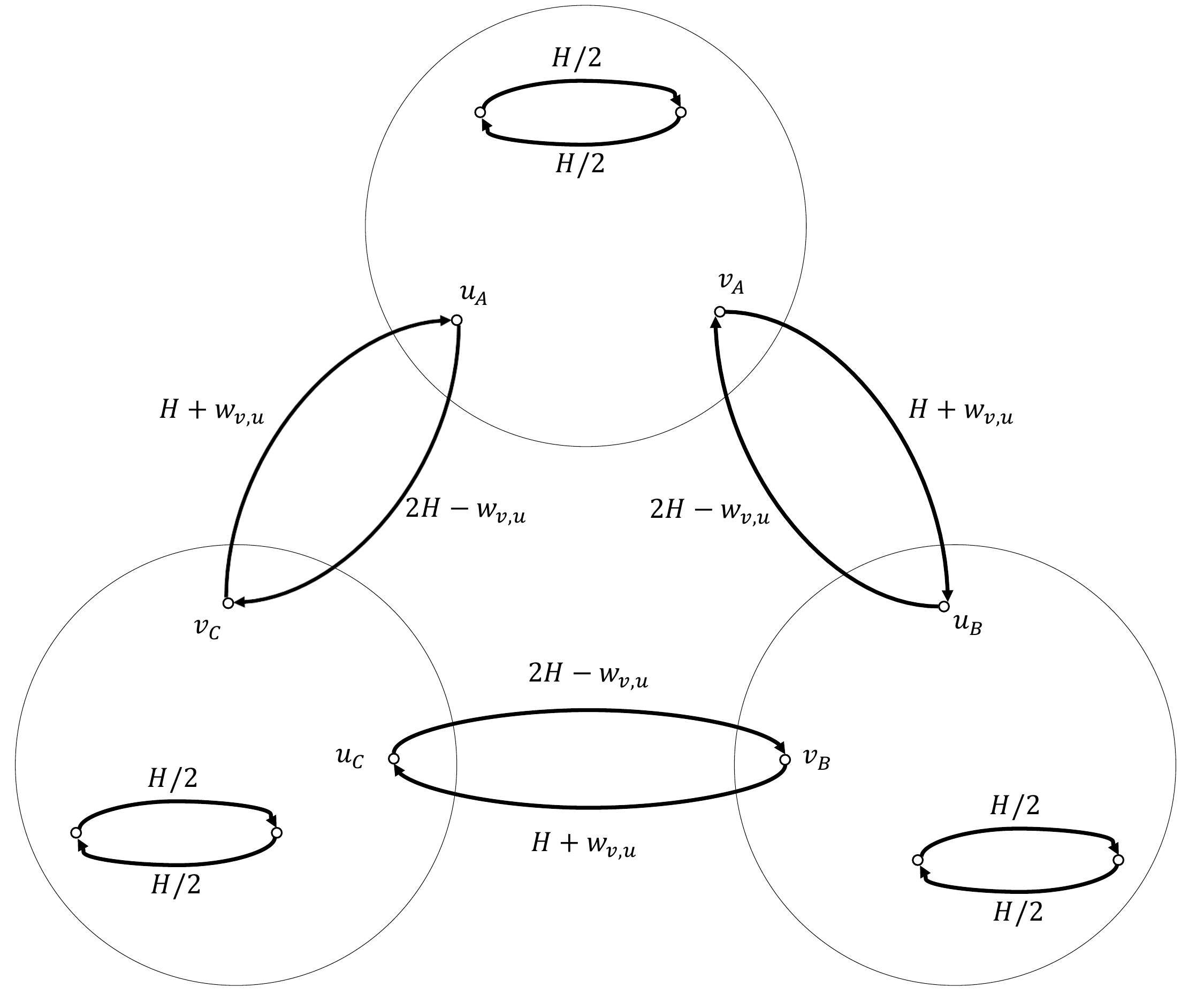}
		\caption{Constructing a symmetric graph $G'$ from $G$ in the reduction from \NEG{} to \COMedian.}
		\label{fig:Gp}
	\end{figure}
	
	Since $G'$ is symmetric, we can assume that it has a median in every part. Let $r_A$ be a median in part $A$. The shortest path from $r_A$ to $v_A$ is a direct edge of weight $H/2$. The shortest path from $r_A$ to every $v_B$ is also a direct edge with weight $H+w_{r,v}$. For every $v_C$, the shortest path from $r_A$ is either a direct edge of weight $2H-w_{r,v}$ or a path through an intermediate vertex $u_B$ with total length of $H+w_{r,u}+H+w_{u,v}$. If the latter is smaller than the former, we can imply that $r$, $u$ and $v$ form a negative triangle in $G$: $$H+w_{r,u}+H+w_{u,v}<2H-w_{r,v} \Rightarrow w_{r,u}+w_{u,v}+w_{v,r}<0\,.$$
	
	On the other hand, if the shortest path from $r_A$ to every vertex $v_C$ is the direct edge $(r_A,v_C)$ then using a similar inequality, it can be shown that $r_A$ does not contribute to any negative triangle in $G$. In this case, all vertices of $G'$ have a fixed summation of distances from all other vertices. Let $sum(v)$ denote such a summation for a vertex $v$. Below, we formulate this value only for vertices in part $A$, because based on the symmetricity of $G'$, the value of $sum(v)$ can be determined via the same formulas for the vertices in parts $B$ and $C$.
	
	
	\begin{align}
	\forall v_A\in G': sum(v_A) &= \sum_{x_A\in A\backslash\{v_A\}}\big(H/2\big)+\sum_{x_B\in B}\big(H+w_{v,x}\big)+\sum_{x_C\in C}\big(2H-w_{v,x}\big) \nonumber \\ 
	&=(n-1)H/2+3nH \label{gg1}
	\end{align}
	
	According to Equation \ref{gg1}, we only need to construct $G'$ as above and see if all vertices are medians, and if $sum(v)=(n-1)H/2+3nH$ for every $v\in V(G')$\footnote{It suffices to check this value just for one vertex because now we know all vertices are median.}. If these two hold, then $G$ is free of negative triangles. Otherwise, there exists a median $r_A$ in $G'$ with $sum(r_A)<(n-1)H/2+3nH$ indicating the existence of a negative triangle in $G$.
\end{proof}

The following theorem follows directly from Lemma \ref{lemlemeto}.
\begin{theorem}
	$\NEG{} \subless \COMedian{}$.
\end{theorem}
\subsection{\Diameter{} to \COAPSPV{}}\label{sec3}
In this section, we provide a subcubic reduction from \Diameter{} to the \COAPSPV{}.
\begin{lemma}\label{lemlemlemeto}
	Given an $\otilda(T(n))$ time algorithm for \COAPSPV{}, where $T(n)$ is polynomial in $n$, there exists an $\otilda(T(n)+n^2)$ time algorithm for \Diameter{}.
\end{lemma}

\begin{proof}
	The outline of the proof is as follows. First, we show an algorithm for finding the solution of the \COAPSPV{} problem can be used as a black box for determining whether the diameter of a graph is greater than or equal to an integer number $d$. Then we run a binary search on $d$ to find the exact diameter of the graph. In the rest, we show how we can determine if the diameter of $G$ is at least some given value $d$.
	
	We construct a graph $G'$ from $G$ by taking all the vertices and edges of $G$ and adding an additional edge from every vertex of $G$ to every other vertex with weight $d$. By taking the minimum, multiple edges of $G'$ can become simple edges. With this construction, the diameter of $G'$ is at most $d$ since there exists a shortcut of weight $d$ between every two vertices. Moreover, if the diameter of $G'$ is exactly $d$, it means there are two vertices $x$ and $y$ in $G$ such that the distance of $y$ from $x$ is at least $d$. Otherwise, the distance of every vertex of $G$ from every other vertex is at most $d-1$. Thus, the diameter of $G$ is more than or equal to $d$ if and only if there exists a pair $(x,y)$ of vertices in $G'$ such that distance of $y$ from $x$ is precisely $d$. Let $D$ be an $n \times n$ matrix such that the entries on the diagonal are $0$ and all of the other entries are equal to $d$. If we give $G'$ and $D$ as inputs to the algorithm for $\COAPSPV$, it will report if any index of $D$ represents the true distance between the corresponding vertices in $G'$, and hence we can determine if the distance between any two vertices of $G'$ is exactly $d$ which is equivalent to $G$ having a diameter of no less than $d$.
\end{proof}

The following theorem follows directly from Lemma \ref{lemlemlemeto}.
\begin{theorem}
$\Diameter{} \subless \COAPSPV{}$.
\end{theorem}
\subsection{\CONEG{} to \Diameter{}}\label{sec4}
In this section, we provide a subcubic reduction from \CONEG{} to \Diameter{}.

\begin{lemma}\label{lemeto}
	Given an $\otilda(T(n))$ time algorithm for \Diameter{}, where $T(n)$ is polynomial in $n$, there exists an $\otilda(T(n)+n^2)$ time algorithm for \CONEG{}.
\end{lemma}

\begin{proof}\sloppy
	Since directed \Diameter{} is harder than its undirected version, we reduce \CONEG{} to undirected \Diameter{}. For every graph $G$, we create an undirected graph $G'$ with six times as many vertices. More precisely, $V(G')$ consists of six parts $A$, $B$, $C$, $D$, $X$, and $Y$. For every vertex $v \in V(G)$, we put vertices $v_A$, $v_B$, $v_C$, $v_D$, $v_X$, and $v_Y$ in parts $A$, $B$, $C$, $D$, $X$, and $Y$, respectively. Moreover, for every edge from a vertex $u$ to a vertex $v$ with weight $w$ in $E(G)$ we draw an edge from $u_A$ to $v_B$, $u_B$ to $v_C$, and from $u_C$ to $v_D$ with weight $w+H$ where $H=10M$. Furthermore, we add an edge from every $v_X$ to $v_A$ with weight $H$. Similarly, we draw an edge from every vertex $v_D$ to $v_Y$ with weight $H$. Finally, for every $u \neq v$ we add an edge from $u_A$ to $v_D$ with weight $0$.
	
	In the following, we show $G$ has a vertex $u$ which does not take part in any negative triangle if and only if the diameter of $G'$ is at least $5H$. Note that due to the construction of $G'$, the diameter of the graph is always the distance between a vertex of part $X$ to a vertex of part $Y$. Since we put an edge of weight $0$ from $u_A$ to $v_D$ for every $u \neq v$, the distance from every vertex $u_X$ to every vertex $v_Y$ is at most $3H$ for $u \neq v$. However, the distance between every vertex $v_X$ to $v_Y$ is more than $3H$. Therefore, the diameter of the graph is always from a vertex $v_X$ to a vertex $v_Y$. Note that, the distance between a vertex $v_A$ to $v_D$ is equal to the weight of the minimum weight triangle in $G$ that contains $v$ plus $3H$. Thus, the diameter of the graph is at least $5H$ if and only if there exists a vertex $v$ in $G$ which lies in no negative triangle.
	
	All that remains is to find a vertex which does not contribute to any negative triangle if such a vertex exists.
	Lemma \ref{mosavian} shows given a subcubic algorithm for \Diameter{} that only finds the length of the diameter we can obtain a subcubic algorithm that also finds the endpoints of the diameter.
	Therefore, we can find two vertices $v_X$ and $v_Y$ such that their distance is equal to the diameter of the graph in time $\otilda(n^{3-\delta})$ for some constant $\delta$. If the distance between $v_A$ and $v_D$ is at least $3H$, we can report $v$ as a vertex which does not contribute to any negative triangle; otherwise, every vertex of $G$ contributes to a negative triangle.
\end{proof}

Theorem \ref{thmthm} follows directly from Lemma \ref{lemeto}.
\begin{theorem}\label{thmthm}
$\CONEG{} \subless \Diameter{}$.
\end{theorem}

\subsection{\CONEG{} to \NEG{}}\label{sec5}
In this section, we reduce \CONEG{} to \NEG{}.
\begin{lemma}\label{lem:CoNeg2APSP}
	Given an $\otilda(T(n))$ time algorithm for \APSP{}, where $T(n)$ is polynomial in $n$, there exists an $\otilda(T(n)+n)$ time algorithm for \CONEG{}.
\end{lemma}

The proof of this lemma is in Appendix \ref{conegnegapp}.
The following theorem follows directly from Lemma \ref{lem:CoNeg2APSP} and $\APSP \subequal \NEG{}$.
\begin{theorem}\label{thm:CoNeg2Neg}
	$\CONEG \subless \NEG{}$.
\end{theorem}



\bibliographystyle{plain}
\bibliography{subcubic}

\begin{thebibliography}{10}

\bibitem{valiant}
Amir Abboud, Arturs Backurs, and Virginia~Vassilevska Williams.
\newblock If the current clique algorithms are optimal, so is valiant's parser.
\newblock In {\em FOCS}, pages 98--117. IEEE, 2015.

\bibitem{lcs}
Amir Abboud, Arturs Backurs, and Virginia~Vassilevska Williams.
\newblock Tight hardness results for {LCS} and other sequence similarity
  measures.
\newblock In {\em FOCS}, pages 59--78. IEEE, 2015.

\bibitem{soda}
Amir Abboud, Fabrizio Grandoni, and Virginia~Vassilevska Williams.
\newblock Subcubic equivalences between graph centrality problems, {APSP} and
  diameter.
\newblock In {\em SODA}, pages 1681--1697. SIAM, 2015.

\bibitem{abboud2014popular}
Amir Abboud and Virginia~Vassilevska Williams.
\newblock Popular conjectures imply strong lower bounds for dynamic problems.
\newblock In {\em FOCS}, pages 434--443. IEEE, 2014.

\bibitem{soda2}
Amir Abboud, Virginia~Vassilevska Williams, and Joshua Wang.
\newblock Approximation and fixed parameter subquadratic algorithms for radius
  and diameter in sparse graphs.
\newblock In {\em SODA}, pages 377--391. SIAM, 2016.

\bibitem{aingworth1999fast}
Donald Aingworth, Chandra Chekuri, Piotr Indyk, and Rajeev Motwani.
\newblock Fast estimation of diameter and shortest paths (without matrix
  multiplication).
\newblock {\em SIAM Journal on Computing}, 28(4):1167--1181, 1999.

\bibitem{editdistance}
Arturs Backurs and Piotr Indyk.
\newblock Edit distance cannot be computed in strongly subquadratic time
  (unless {SETH} is false).
\newblock In {\em STOC}, pages 51--58. ACM, 2015.

\bibitem{ted}
Karl Bringmann, Pawe\l{} Gawrychowski, Shay Mozes, and Oren Weimann.
\newblock Tree edit distance cannot be computed in strongly subcubic time
  (unless {APSP} can).
\newblock In {\em SODA}, pages 1190--1206. SIAM, 2018.

\bibitem{chan2010more}
Timothy~M Chan.
\newblock More algorithms for all-pairs shortest paths in weighted graphs.
\newblock {\em SIAM Journal on Computing}, 39(5):2075--2089, 2010.

\bibitem{floyd}
Robert~W Floyd.
\newblock Algorithm 97: shortest path.
\newblock {\em Communications of the ACM}, 5(6):345, 1962.

\bibitem{fredman1976new}
Michael~L Fredman.
\newblock New bounds on the complexity of the shortest path problem.
\newblock {\em SIAM Journal on Computing}, 5(1):83--89, 1976.

\bibitem{hakimi1964optimum}
S~Louis Hakimi.
\newblock Optimum locations of switching centers and the absolute centers and
  medians of a graph.
\newblock {\em Operations research}, 12(3):450--459, 1964.

\bibitem{han2012n}
Yijie Han and Tadao Takaoka.
\newblock An {$O(n^3 \log \log n/\log^2 n)$} time algorithm for all pairs
  shortest paths.
\newblock In {\em SWAT}, pages 131--141. Springer, 2012.

\bibitem{impagliazzo1999complexity}
Russell Impagliazzo and Ramamohan Paturi.
\newblock On the complexity of {k-SAT}.
\newblock In {\em CCC}, pages 237--240. IEEE, 1999.

\bibitem{impagliazzo1998problems}
Russell Impagliazzo, Ramamohan Paturi, and Francis Zane.
\newblock Which problems have strongly exponential complexity?
\newblock In {\em FOCS}, pages 653--662. IEEE, 1998.

\bibitem{soda3}
Andrea Lincoln, Virginia~Vassilevska Williams, and Ryan Williams.
\newblock Tight hardness for shortest cycles and paths in sparse graphs.
\newblock In {\em SODA}, pages 1236--1252. SIAM, 2018.

\bibitem{patrascu2010towards}
Mihai P{\v{a}}tra{\c{s}}cu.
\newblock Towards polynomial lower bounds for dynamic problems.
\newblock In {\em STOC}, pages 603--610. ACM, 2010.

\bibitem{roditty2013fast}
Liam Roditty and Virginia Vassilevska~Williams.
\newblock Fast approximation algorithms for the diameter and radius of sparse
  graphs.
\newblock In {\em STOC}, pages 515--524. ACM, 2013.

\bibitem{shoshan1999all}
Avi Shoshan and Uri Zwick.
\newblock All pairs shortest paths in undirected graphs with integer weights.
\newblock In {\em FOCS}, pages 605--614. IEEE, 1999.

\bibitem{takaoka2005n3loglogn}
Tadao Takaoka.
\newblock An {$O(n^3\log\log n/\log n)$} time algorithm for the all-pairs
  shortest path problem.
\newblock {\em Information Processing Letters}, 96(5):155--161, 2005.

\bibitem{warshall}
Stephen Warshall.
\newblock A theorem on boolean matrices.
\newblock {\em Journal of the ACM}, 9(1):11--12, 1962.

\bibitem{williams2014faster}
Ryan Williams.
\newblock Faster all-pairs shortest paths via circuit complexity.
\newblock In {\em STOC}, pages 664--673. ACM, 2014.

\bibitem{ipecsurvey}
Virginia~Vassilevska Williams.
\newblock Hardness of easy problems: Basing hardness on popular conjectures
  such as the strong exponential time hypothesis (invited talk).
\newblock In {\em IPEC}, pages 17--29. Dagstuhl, 2015.

\bibitem{focs}
Virginia~Vassilevska Williams and Ryan Williams.
\newblock Subcubic equivalences between path, matrix and triangle problems.
\newblock In {\em FOCS}, pages 645--654. IEEE, 2010.

\bibitem{zwick2002all}
Uri Zwick.
\newblock All pairs shortest paths using bridging sets and rectangular matrix
  multiplication.
\newblock {\em Journal of the ACM}, 49(3):289--317, 2002.

\bibitem{zwick2005slightly}
Uri Zwick.
\newblock A slightly improved sub-cubic algorithm for the all pairs shortest
  paths problem with real edge lengths.
\newblock In {\em Algorithms and Computation}, pages 921--932. Springer, 2005.

\end{thebibliography}

\newpage

\appendix
\section{Equivalence of the Vertex Version and the Numerical Version}\label{trivialapp}
\begin{proofof}{Lemma \ref{mosavian}}
	For each of these three problems, the reduction from the numerical version to the variant in which the output of the problem is a vertex is trivial. This holds since having an optimal solution of the vertex version, we can simply run a single-source-shortest-path algorithm and find the numerical solution of the problem in time $O(n^2)$.
	
	To reduce the vertex version to the numerical version, we assume there is an algorithm that solves the numerical version and will access the algorithm as a black box $O(\log n )$ times. Note that this keeps the reduction subcubic since the number of accesses is less than $n^\epsilon$ for any $\epsilon > 0$. The overall idea is the same for all of the three problems. We first show how to use the solver of the numerical version to see whether or not a set $S\subseteq V$ contains a solution of the vertex version. Then everything boils down to a binary search: Beginning from a set of vertices $S=V(G)$, at each step we divide $S$ into two subsets of size fairly equal $S_1$ and $S_2$ and search for the solution of the vertex version in either $S_1$ or $S_2$. This cuts the size of the search space in half at every step and finally finds the desired vertex in at most $\lceil \log(n) \rceil$ steps.
	
	In the following, for each of the three problems, we show how to use a solver of the numerical version to see whether or not there exists a solution of the vertex version in $S$.
	\begin{itemize}
		\item \Diameter{}: Suppose the diameter of $G$ is equal to $d$. We construct $G'$ from $G$ by multiplying all edge weights by three and adding a dummy vertex $x'$ for each $x\in S$ and connecting them with two edges of weight one, an edge from $x$ to $x'$ and an edge from $x'$ to $x$. Let $d'$ denote the diameter of $G'$. If no vertex in $S$ is an endpoint of a diameter of $G$ then $d'$ will be equal to $3d$. Otherwise, $d'$ will be $3d+1$ or $3d+2$.
				
		\item \Radius{}: Suppose the radius of $G$ is $r$. We construct $G'$ from $G$ by adding a dummy node $x$ and connecting all vertices of $S$ to $x$ with edges of weight $r$. Let $r'$ be the radius of $G'$. If $r'>r$ then all centers of $G$ are in $V\backslash S$, since they need to reach $x$ through $S$ in $G'$. Otherwise, there exists a center of $G$ in $S$ and $r'=r$.
		\item \Median{}: Suppose $m$ is the value of the median of $G$. We construct $G'$ from $G$ by adding a dummy vertex $x$. Let's $Q$ be a large number. We connect all vertices of $S$ to $x$ with edges of weight $Q$. Moreover, we connect all vertices outside of $S$ to $x$ with edges of weight $Q+1$. Let us use $m'$ to denote the value of median in $G'$. If there exists a median vertex of $G$ in $S$, then that vertex can be a median vertex of $G'$, too. In this case, $m'=m+Q$. If no such a vertex exists, then all medians of $G'$ are outside of $S$ and $m'=m+Q+1$.
	\end{itemize}
	\begin{figure}[h]
		\includegraphics[width=4cm]{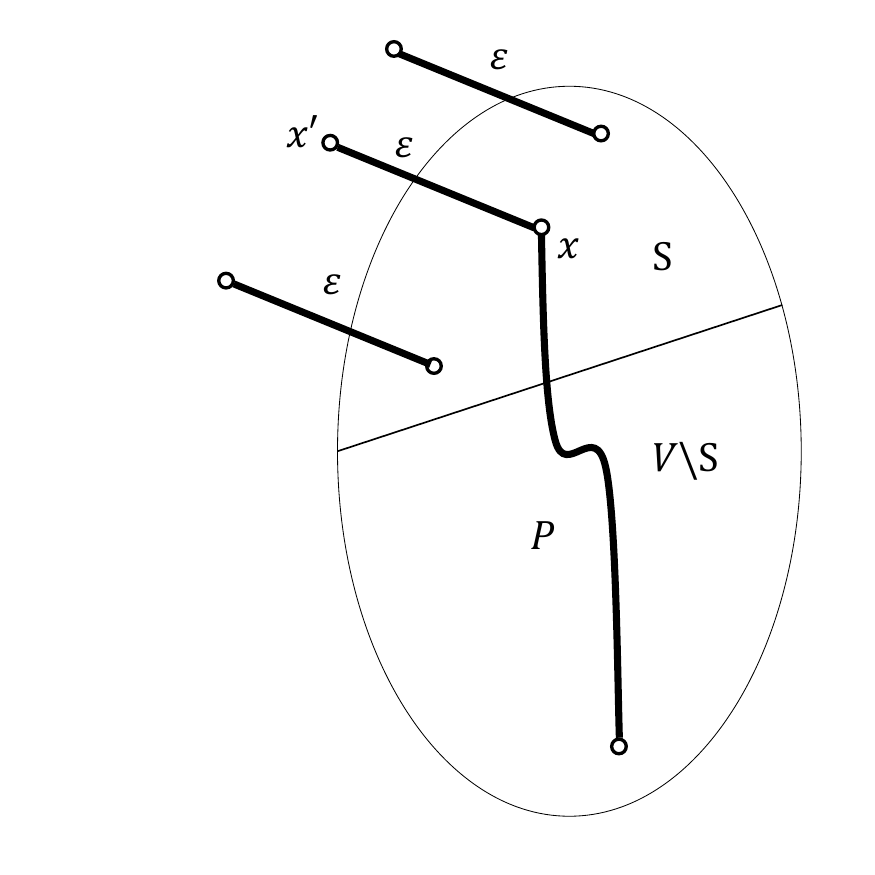}
		\includegraphics[width=4cm]{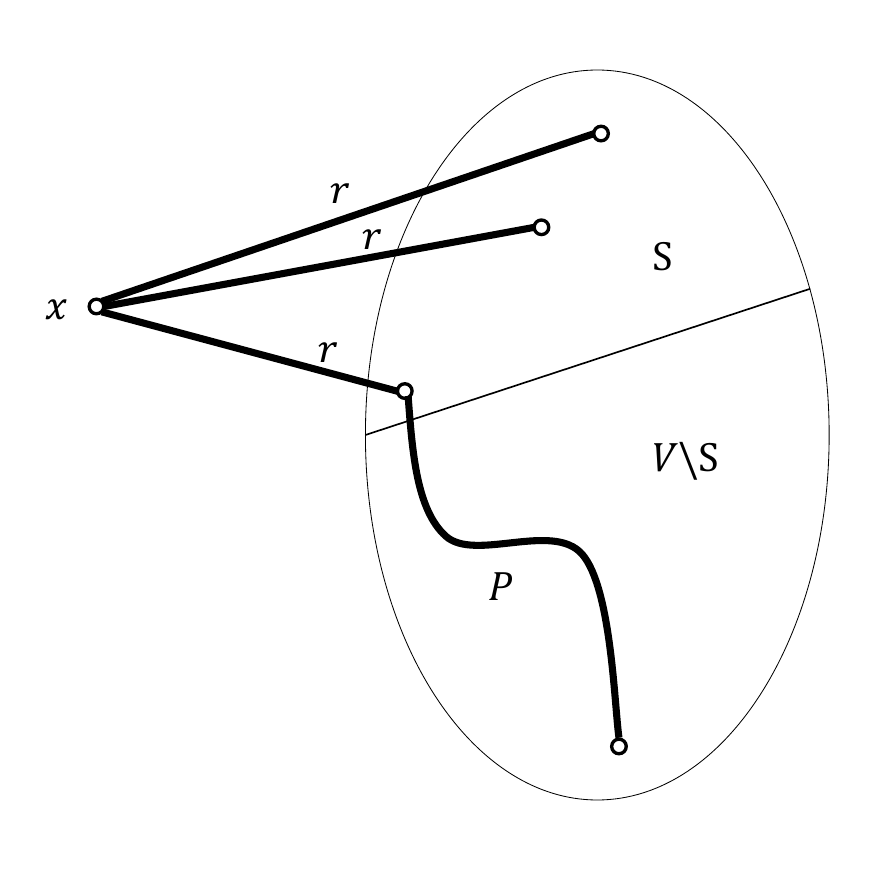}
		\includegraphics[width=4cm]{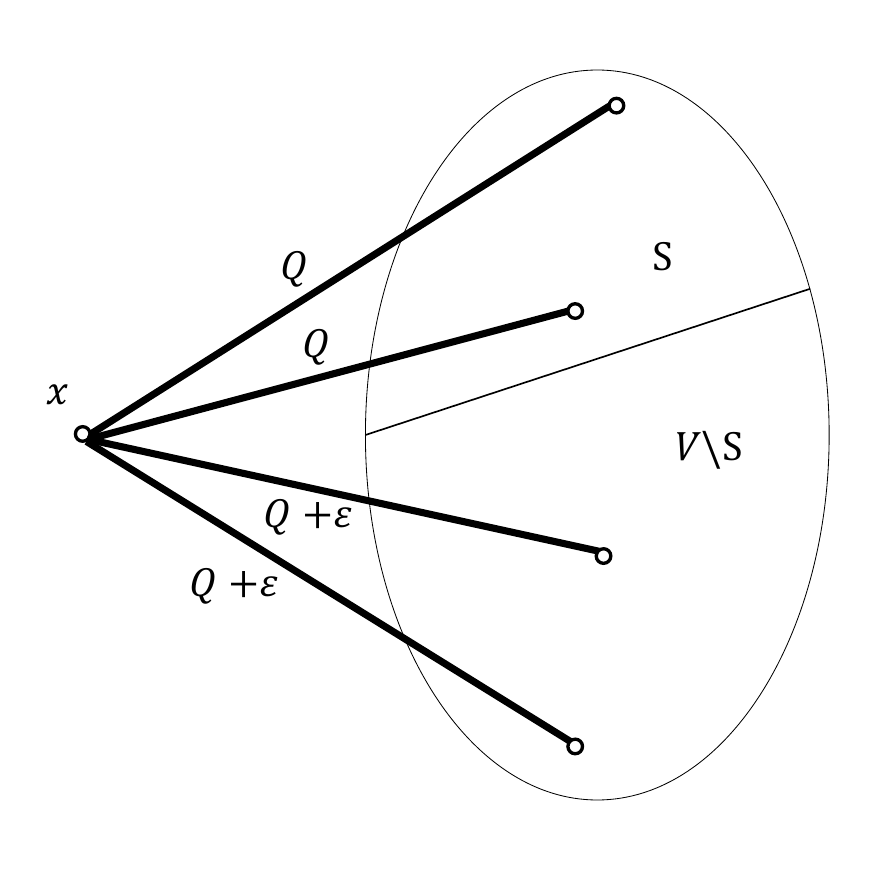}
		\caption{(i) Left figure shows the reduction for \Diameter{}, (ii) the figure in the middle illustrates the reduction for \Radius{}, and (iii) the figure on the right shows the reduction for \Median{}. }
	\end{figure}
\end{proofof}
\section{Reducing \CONEG{} to \NEG{}}\label{conegnegapp}
\begin{proofof}{Lemma \ref{lem:CoNeg2APSP}}
	For every graph $G$, we create a graph $G'$ with four times as many vertices. To be more precise, $V(G')$ consists of four parts $A$, $B$, $C$, and $D$. For every vertex $v \in V(G)$, we put vertices $v_A$, $v_B$, $v_C$, and $v_D$ in parts $A$, $B$, $C$, and $D$, respectively. Furthermore, for an edge from a vertex $u$ to a vertex $v$ with weight $w$ in $E(G)$ we add an edge from $u_A$ to $v_B$, $u_B$ to $v_C$, and from $u_C$ to $v_D$ with weight $w+M$. A vertex $v \in V(G)$ is in a negative triangle if and only if the distance between $v_A$ and $v_D$ is less than $3M$. To this end, we run the $\otilda(T(n))$ time algorithm for \APSP{} on $G'$ and check the distance between $v_A$ and $v_D$ for each $v \in V(G)$ and see whether a vertex exists that does not belong to any negative triangle. Therefore, we can solve \CONEG{} in $\otilda(T(n)+n)$ time.
\end{proofof}

\end{document}